\documentclass[a4paper,oneside,11pt]{article}

\usepackage{array,ifthen,float,a4wide,amsfonts,subfigure,algorithm,graphicx,amssymb,amsmath,latexsym}

\newtheorem{definition}{Definition}
\newtheorem{theorem}{Theorem}

\newtheorem{lemma}{Lemma}

\floatstyle{ruled}
\newfloat{Algorithm}{thp}{lop}[section]

\newenvironment{proof}{\noindent{\textbf{Proof:}}}{\hfill$\Box$}

\begin{document}

\title{Limiting Byzantine Influence in Multihop Asynchronous Networks}

\author{Alexandre Maurer and S\'{e}bastien Tixeuil\\
UPMC Sorbonne Universit\'{e}s, Paris, France\\
Email: Alexandre.Maurer@lip6.fr, Sebastien.Tixeuil@lip6.fr
}

\maketitle

\begin{abstract}
We consider the problem of reliably broadcasting information in a multihop asynchronous network that is subject to Byzantine failures. That is, some nodes of the network can exhibit arbitrary (and potentially malicious) behavior. Existing solutions provide deterministic guarantees for broadcasting between \emph{all} correct nodes, but require that the communication network is highly-connected (typically, $2k+1$ connectivity is required, where $k$ is the total number of Byzantine nodes in the network).

In this paper, we investigate the possibility of Byzantine tolerant reliable broadcast between \emph{most} correct nodes in low-connectivity networks (typically, networks with constant connectivity). In more details, we propose a new broadcast protocol that is specifically designed for low-connectivity networks. We provide sufficient conditions for correct nodes using our protocol to reliably communicate despite Byzantine participants. We present experimental results that show that our approach is especially effective in low-connectivity networks when Byzantine nodes are randomly distributed.
\end{abstract}

\section{Introduction}

In this paper, we revisit the problem of reliably broadcasting information in an arbitrary shaped network that is subject to Byzantine failures. As distributed systems and networks grow larger, faults and attacks are more likely to appear and resiliency to those faults needs to be adressed in the very early design stages of the protocols that target those systems. One of the strongest fault models is \emph{Byzantine} \cite{LSP82j}: the faulty node behaves arbitrarily. This model encompasses a rich set of fault scenarios. Moreover, Byzantine fault tolerance has security implications, as the behavior of an intruder can be modeled as Byzantine. However, in most studies to date, Byzantine faults are considered in completely connected networks.

One approach to deal with Byzantine faults is by enabling the nodes to use cryptographic operations such as digital signatures or certificates. This limits the power of a Byzantine node as a correct node can verify the validity of received information and authenticate the sender across multiple hops. However, this option may not be available. For example, the nodes may not have enough resources to manipulate digital signatures. Moreover, cryptographic operations implicitly assume the presence of a trusted infrastructure: secure channels to a key server or a public key infrastructure. Establishing and maintaining such infrastructure in the presence of Byzantine faults may be problematic.

Another way to limit the power of a Byzantine process is to assume synchrony: all processes proceed in lock-step. Indeed, if a process is required to send a message with each pulse, a Byzantine process cannot refuse to send a message without being detected. However, the synchrony assumption may be too restrictive for practical systems.

\paragraph{Related works} Many recent Byzantine-robust protocols make use of \emph{cryptography} (see \cite{CL99c,DFS05c} and references herein) to contain the influence of Byzantine nodes. As previously stated, this requires a trusted infrastructure that we do not assume. 

Cryptography-free Byzantine failures have first been studied in completely connected networks~\cite{LSP82j,AW98b,MMR03j,MRRS01c,MS03j}: a node can directly communicate with any other node. If the underlying network is not completely connected (as are most networks), such a setting raises the problem to reliably transmit information between nodes that are not direct neighbors. 

In practice, broadcasting a message in the network requires to rely upon other nodes. Dolev~\cite{D82j} considers Byzantine agreement on arbitrary graphs. He states that, for agreement in the presence of up to $k$ Byzantine nodes, it is necessary and sufficient that the network is $(2k+1)$-connected, and that the number of nodes in the system is at least $3k+1$. Also, this solution assumes that the underlying graph is known to every node, and that nodes are scheduled according to the synchronous execution model. Nesterenko and Tixeuil~\cite{NT09j} relax both requirements (graph is unknown and scheduling is asynchronous) yet retain $2k+1$ connectivity for resilience and $k+1$ connectivity for detection (that is, detecting that there is at least one Byzantine node in the network). In a low-connectivity network such as a torus (where nodes have degree at most four), both approaches can cope only with a single Byzantine node, independently of the torus size. 

Byzantine resilient broadcast was recently investigated in the context of \emph{radio networks}: each node is a robot or a sensor with a physical position. A node can only communicate with nodes that are located within a certain radius, called neighbors. Broadcast protocols have been proposed \cite{K04c,BV05c} for nodes organized on a grid (the wireless medium typically induces much more than four neighbors per node, otherwise the broadcast does not work). Both approaches are based on a local voting system, and perform correctly if every node has less than a $1/4\pi$ fraction of Byzantine neighbors. This criterion was later generalized \cite{PP05j} to other topologies, assuming that each node knows the exact topology. Again, in low-connectivity networks, the local constraint on the proportion of Byzantine nodes in any neighborhood may be difficult to assess.

A notable class of algorithms tolerates Byzantine faults with either space~\cite{MT07j,NA02c,SOM05c} or time~\cite{MT06cb,DMT11cb,DMT11j,DMT10cd,DMT10ca} locality. Yet, the emphasis of space local algorithms is on containing the fault as close to its source as possible. This is only applicable to the problems where the information from remote nodes is unimportant such as vertex coloring, link coloring or dining philosophers. Also, time local algorithms presented so far can hold at most one Byzantine node and are not able to mask the effect of Byzantine actions. Thus, the local containment approach is not applicable to reliable broadcast.

\paragraph{Our contribution} All aforementioned results rely on a strong \emph{connectivity} of the communication graph, and on Byzantine proportions assumptions in the network. In other words, tolerating more Byzantine nodes requires an increase of the degree of each node, which can be a heavy constraint on large networks such as a peer-to-peer overlays.

In this paper, we introduce the idea to trade the perfectly reliable communication between correct nodes for the ability to support low-connectivity communication graphs with many Byzantine nodes. Informally, we accept that a small minority of correct nodes are denied reliable communication, provided that a large majority of correct nodes can reliably communicate. Such a loss in safety guarantees looks acceptable to us since, by hypothesis, many Byzantine nodes (with arbitrary behavior) are already present in the system. Also, our results demonstrate that this tradeoff permits to tolerate a high number of Byzantine nodes, even in constant connectivity networks such as grids or torus. Yet, experimental results show that our scheme preserves the communication capabilities of a huge majority of correct processes.

Our approach is not based on voting proportions, but on \emph{control zones} and \emph{authorizations}. Intuitively, control zones act as filters in the network: they limit the diffusion of Byzantine messages.

The sequel of the paper is organized as follows.
In Section~\ref{secdes}, we define a new broadcast protocol based on \emph{control zones}.
In Section~\ref{secthm}, we prove sufficient conditions to achieve reliable communication in a particular subgraph of the network, using our protocol.
In Section~\ref{secexp}, we provide an experimental evaluation of our protocol on \emph{torus} and \emph{grid} shaped networks, with randomly distributed Byzantine failures. Using the sufficient conditions of Section~\ref{secthm}, we evaluate the probability for two randomly choosen nodes to communicate reliably.

\section{Description of the protocol}

\label{secdes}
In this section, we give an informal description of our protocol, set the formal background, and describe how the protocol is locally executed.

\subsection{Informal Description}

\label{informal}

The network is described by a set of processes, called \emph{nodes}. Some pairs of nodes are linked by a \emph{canal}, and can send messages to each other: we call them \emph{neighbors}. The network is \emph{asynchronous}: the nodes can send and receive messages at any time. Our only hypothesis is that any message that is sent is eventually received.

A node may want to broadcast a specific information $m_{0}$ to the the network. For instance, in a sensors network, $m_{0}$ can be a temperature; in a mobile robots network, $m_{0}$ can be the position of the current robot; etc. In a network where all nodes are correct, the following would happen:
\label{classical}
\begin{itemize}
\item A given node $p$ sends a message containing the couple $(p,m_{0})$ to its neighbors,
\item the neighbors of $p$ send $(p,m_{0})$ to their neighbors,
\item and so forth, until every node receives $(p,m_{0})$. Then the entire network knows that $p$ broadcasted the information $m_{0}$.
\end{itemize}

In our setting however, some nodes can be Byzantine and send arbitrary messages. Those messages are potentially malicious. For instance, a Byzantine node can send $(p,m_{1})$, with $m_{1} \neq m_{0}$, to make the network believe that $p$ broadcasted the information $m_{1}$. Therefore, one single Byzantine node can lie about the information of every node, and then deceive the whole network.

To limit the action of Byzantine nodes, we define a set of \emph{control zones}. A control zone is defind by:
\begin{itemize}
\item Its \emph{core}, an arbitrary set of nodes.
\item Its \emph{border}, a node-cut isolating the core from the rest of the network.

\begin{figure}
\begin{center}
\includegraphics[width=7cm]{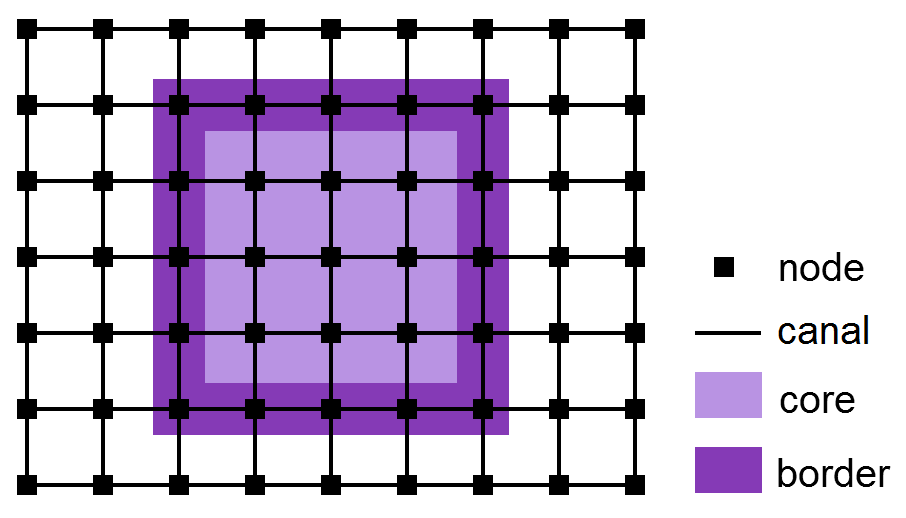}
\caption{Example of control zone} 
\label{fig:ctrzone}
\end{center}
\end{figure}

\end{itemize}
An example of control zone is given in Fig.~\ref{fig:ctrzone}.
The important point is that messages must pass through the \emph{border} to access the \emph{core}.

Here is the main idea of the protocol:
\begin{itemize}
\item When a message enters the \emph{core} of a control zone, an \emph{authorization} is broadcasted on its \emph{border}.
\item When the same message wants to exit the \emph{core}, this \emph{authorization} is required.
\end{itemize}
This mechanism does not disturb the broadcasting of correct messages.

Now, suppose that a Byzantine node is in the core of the control zone, and sends a lying message $(p,m_{1})$, whereas $p$ is \emph{not} in the core of the control zone. Then, this message never gets the authorization to exit the core, as it never entered it. This is illustrated in Fig.~\ref{fig:mainidea}.

\begin{figure}
\begin{center}
\includegraphics[width=8cm]{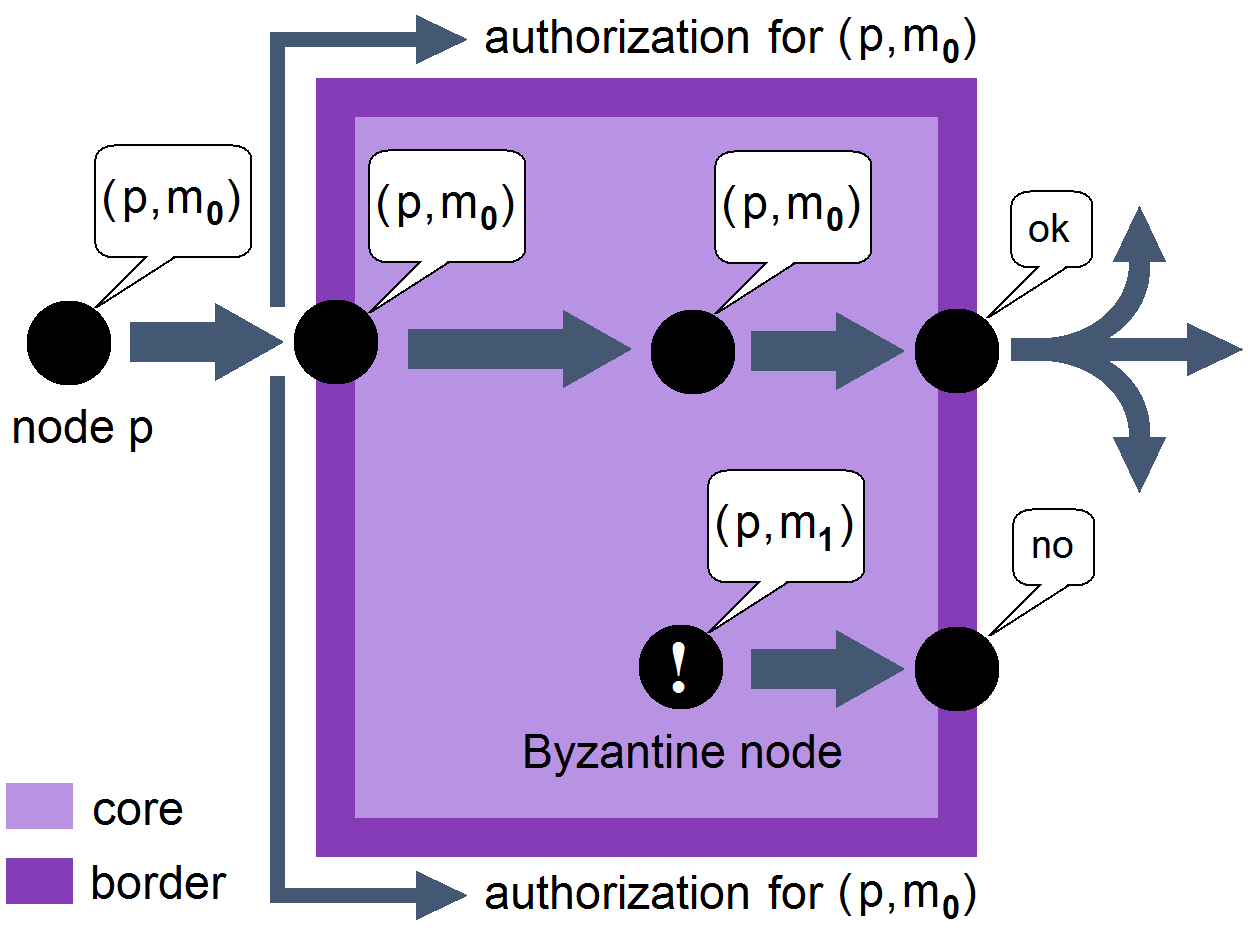}
\caption{Principle of a control zone} 
\label{fig:mainidea}
\end{center}
\end{figure}

Intuitively, this mechanism of control zones enables to limit the broadcasting of Byzantine messages. The underlying idea is to define a lot of control zones on the network, intersecting each other, in order to minimize the broadcasting of Byzantine messages. Then eventually, under certain conditions, we could achieve reliable broadcast in a specific part of the network.

\subsection{Definitions and Hypothesis}

\label{notadef}

\subsubsection{Network Model}

Let $(G,E)$ be a non-oriented graph representing the topology of the network.
\begin{itemize}
\item $G$ denotes the \emph{nodes} of the network.
\item $E$ denotes the \emph{neighborhood} relationships: two given nodes are or are not neighbors. A node can only send messages to its neighbors.
\end{itemize}

Let $Corr$ be the set of \emph{correct} nodes. These nodes follow the protocol described further. We assume that the other nodes are totally unpredictable (\emph{Byzantine}).

We assume asynchronous message passing: any message sent is eventually received, but it can be at any time. We assume that, in an infinite execution, any process is activated inifinitely often. We make no hypothesis on the order of activation of the processes.
We assume local topology knowledge: each node knows the identifiers of its neighbors. Therefore, its direct neighbors cannot lie on their identity when sending a message. Also, a  slightly stronger topology knowledge is required (see definition of $myCtr$ in~\ref{mymyctr}).

\subsubsection{Control zones}

\label{mydefctr}

A set of nodes $S$ is \emph{connected} if, for any nodes $p$ and $q$ in $S$, there exists a chain of neighborhood relationships linking $p$ and $q$. A \emph{node-cut} is a set of nodes $C$ such that $G \backslash C$ is disconnected.

\begin{definition}[Control zone]
\label{defctrz}
A \emph{control zone} is a pair $(Core, Border)$ of disjoint, connected node sets, such that $Border$ is a node-cut isolating $Core$ from the rest of the network.
\end{definition}

We denote the core and the border of a control zone $z$ by $core(z)$ and $border(z)$.

Before running the protocol, we choose an arbitrary set $Ctr$ of control zones. These are the control zones used in the protocol.

\label{defmyctr}

\subsubsection{Messages formalism}
In the protocol, two types of messages can be exchanged:

\begin{itemize}
\item \emph{Standard messages}, of the form $(s,m)$: a message claiming that the node $s$ (\emph{source}) initially sent the information $m$.
\item \emph{Authorization messages}, of the form $(s,m,z)$: a message authorizing the \emph{standard message} $(s,m)$ to exit the control zone $z$.
\end{itemize}

\subsubsection{Attributes of the correct nodes}
\label{mymyctr}

Each correct node $p$ posesses two static attributes:
\begin{itemize}
\item  $m_{0}$: the information that $p$ wants to broadcast.
\item  $myCtr$: set of control zones $z \in Ctr$ such that $p \in border(z)$. We assume that, for each zone $z \in myCtr$, $p$ knows which nodes belong to $core(z)$ and $border(z)$.
\end{itemize}

It also possesses three dynamic attributes:

\begin{itemize}
\item  $Wait$: set of messages received, but waiting for an authorization (initially empty). When $(s,m,q) \in Wait$, it means that $p$ received the \emph{standard message} $(s,m)$ from a neighbor $q$.

\item $Auth$: set of authorizations received (initially empty).
When $(s,m,z) \in Auth$, it means that $p$ has received the authorization for the \emph{standard message} $(s,m)$ on the control zone $z$.
\item $Acc$: set of accepted messages (initially empty).
When $(s,m) \in Acc$, it means that $p$ has received $(s,m)$ and all the corresponding authorizations, and has sent it to its neighbors.
\end{itemize}

The attribute $X$ of a node $p$ is denoted by $p.X$.

\subsection{Local Execution of the Protocol}

\label{localex}
The protocol is locally executed by each correct node.

Let $p$ be the identifier of the current node executing the protocol. $Wait$, $Auth$, $Acc$, $myCtr$ and $m_{0}$ refer to the corresponding attributes of $p$.

The protocol contains the following actions:
\begin{itemize}
\item \emph{Send} a particular message to the neighbors of $p$.
\item \emph{Add} an element $x$ to a set $X$ ($X := X \cup \{x\}$).
\end{itemize}

The protocol is divided in four sections, executed in specific moments: INIT, ENTER, DIFF and EXIT.

\subsubsection{INIT - Initial broadcast}

Executed initially.

\begin{itemize}
\item Send $(p,m_{0})$ to all neighbors.
\item Add $(p,m_{0})$ to $Acc$.
\item $\forall z \in myCtr$, send $(p,m_{0},z)$ to all neighbors.
\end{itemize}

\subsubsection{ENTER - Message entering control zones}

Executed when a standard message $(s,m)$ is received from a neighbor $q$.
\begin{itemize}
\item If $(s,m) \in Acc$, ignore it.
\item Else, add $(s,m,q)$ to $Wait$.
\end{itemize}

\subsubsection{DIFF - Diffusion of authorizations}

Executed when an authorization message $(s,m,z)$ is received from a neighbor $q$.
\begin{itemize}
\item If $(s,m,z) \in Auth$, or $q \notin border(z)$, ignore it.
\item Else: \begin{itemize}
\item Add $(s,m,z)$ to $Auth$.
\item Send $(s,m,z)$ to all neighbors.
\end{itemize}
\end{itemize}

\subsubsection{EXIT - Message exiting control zones}

Executed when an element $(s,m,q)$ of $Wait$ verifies the following condition: $\forall z \in myCtr$, such that $q \in core(z)$ and $s \notin core(z)$, we have $(s,m,z) \in Auth$.
\begin{itemize}
\item Add $(s,m)$ to $Acc$.
\item Send $(s,m)$ to all neighbors.
\item $\forall z \in myCtr$, send $(s,m,z)$ to all neighbors.
\end{itemize}

\section{Properties of the protocol}

\label{secthm}

In this section, we adopt the point of view of an omniscient observer, knowing the topology of the network and the position of all Byzantine nodes. The following theorems enable to determine sets of nodes satisfying certain properties in \emph{any} execution: \emph{safety}, \emph{communication}, \emph{reliability} (see Definitions~\ref{defsafe}, \ref{defcom}, \ref{defrel}).

\begin{itemize}
\item In Theorem~\ref{thsafe}, we show how to determine the \emph{safe} nodes (who never accept a false message).
\item In Theorem~\ref{thcom}, we show how to construct a \emph{communicating} node set (where all correct messages are received).
\item In Theorem~\ref{threl}, we show that a \emph{safe} and \emph{communicating} node set achieves reliable communication.
\end{itemize} 

Notice that it does not require that any correct node knows the position of the Byzantine nodes: this is just a global vision of the network. These theorems are used in Section \ref{secexp}, to evaluate the performances of the protocol for a given placement of Byzantine nodes.

We also analyze the message complexity of the protocol.

\subsection{Notations and Definitions}

We say that a correct node $p$ \emph{accepts} a message $(s,m)$, when $(s,m)$ is added to the set $p.Acc$.

\begin{definition}[Correct and false messages]
A message $(s,m)$ is \emph{correct} if $m = s.m_{0}$. Else, it is \emph{false}.
\end{definition}

\begin{definition}[Correct path]
\label{defpath}
Let $S$ be a node set. Let $p$ and $q$ be two nodes of $S$.
A \emph{correct path} on $S$ between $p$ and $q$ is a serie $(i_{1},\dots,i_{n})$ of correct nodes of $S$ such that:
\begin{itemize}
\item $i_{1} = p$.
\item $i_{n} = q$.
\item $i_{k}$ and $i_{k+1}$ are neighbors.
\end{itemize}
\end{definition}

\begin{definition}[Safe node set]
\label{defsafe}
A node is \emph{safe} if it never accepts a false message. A node set is safe if all its nodes are safe.
\end{definition}

\begin{definition}[Communicating node set]
\label{defcom}
A set $S$ of correct nodes is \emph{communicating} if, for any nodes $p$ and $q$ of $S$, $q$ eventually accepts $(p,p.m_{0})$.
\end{definition}

\begin{definition}[Reliable node set]
\label{defrel}
A set $S$ of correct nodes is \emph{reliable} if, for any nodes $p$ and $q$ of $S$, $q$ never accepts a false message, and eventually accepts $(p,p.m_{0})$.
\end{definition}

\subsection{Determination of a safe node set}

The following theorem enables, under certain conditions, to determine a safe node set in the network. The condition is the existence of a particular set $Z$ of control zones, such that:
\begin{itemize}
\item The union of the cores and the union of the borders are disjoint.
\item The union of the cores contains all Byzantine nodes.
\end{itemize}
Then, all the nodes outside of $Z$ are safe.
Notice that no correct node needs to determine $Z$: we just have to know that this set exists.

\begin{theorem}[Determination of safe nodes]

\label{thsafe}
If there exists a set $Z$ of control zones $z \in Ctr$ such that:

With $Cores = \cup_{z \in Z}$ $core(z)$
\\With $Borders = \cup_{z \in Z}$ $border(z)$
\begin{itemize}
\item (1) The node sets $Cores$ and $Borders$ are disjoint.
\item (2) All Byzantine nodes are in $Cores$.
\end{itemize}

Then any node $v \notin Cores$ is safe.
\end{theorem}

\begin{proof}
The proof is by contradiction.
Suppose the opposite: let $(s,m)$ be any false message, that is $m \neq s.m_{0}$.
And let $v$ be the first correct node such that:
\begin{itemize}
\item (a) $v \notin Cores$.
\item (b) $v$ accepts $(s,m)$, that is: $v$ is \emph{not} safe.
\end{itemize}

Obviously, $v$ did not accept $(s,m)$ in INIT, as $m \neq s.m_{0}$. So it was in EXIT.
Thus, there exists $(s,m,q) \in v.Wait$ verifying the condition of EXIT.
And the only way for $(s,m,q)$ to have joined $v.Wait$, is that $v$ received $(s,m)$ from $q$ in ENTER. Then, two possibilities:

\begin{itemize}
\item Either $q$ is a correct node, and accepted $(s,m)$ in EXIT. As $v$ is the first node to verify (a) and (b), it implies that $q \in Cores$.
\item Either $q$ is a Byzantine node. Then, according to (2), $q \in Cores$.
\end{itemize}

So, in any case, $q \in Cores$. Therefore, let $z \in Z$ be a control zone such that $q \in core(z)$.
As $v \notin Cores$, $v \notin core(z)$. But $v$ is neighbor of $q \in core(z)$.
So, by definition of a control zone (see Definition~\ref{defctrz}), $v \in border(z)$: otherwise, the border would not be a node-cut isolating the core. 

Then, by definition of $myCtr$ (see \ref{mymyctr}), $z \in v.myCtr$.
As $(s,m,q)$ verifies the condition of EXIT, $z \in v.myCtr$ implies that $(s,m,z) \in v.Auth$.

The only way for $(s,m,z)$ to have joined $v.Auth$, is that $v$ received $(s,m,z)$ in DIFF, from a neighbor in $border(z)$. Let $u$ be the first node of $border(z)$ to send $(s,m,z)$. As $Cores$ and $Borders$ are disjoint, according to (2), $u$ is correct. And $u$ did not send $(s,m,z)$ in DIFF: otherwise, it would not be the first to do so. So it was in EXIT, implying that $u$ accepted $(s,m)$. So $u$ verified (a) and (b) before $v$. This contradiction achieves the proof.
\end{proof}

\subsection{Construction of a communicating node set}

The following theorem enables to construct a communicating set node by node, when it is possible. Let $S$ be a given communicating set, and $v$ a given correct node: then the theorems tells us if $S \cup \{v\}$ is communicating, and so forth. The conditions are the following:
\begin{itemize}
\item The node $v$ has a neighbor in $S$.
\item We have enough \emph{correct paths} (see Definition~\ref{defpath}) to \emph{always} receive the
authorizations required for the communication between $v$ and $S$.
\end{itemize}
To initiate the construction of $S$: simply notice that any correct node $p$ forms a communicating node set $\{p\}$.

\begin{theorem}[Construction of a communicating node set]
\label{thcom}
Let $S$ be a communicating node set.
Let $v$ be a correct node verifying the following conditions:
\begin{itemize}
\item (1) v has a neighbor $u \in S$
\item (2) Let $Z$ be the set of control zones $z \in Ctr$, such that $u \in core(z)$
and $v \in border(z)$. Then $\forall z \in Z$, there exists a correct path on $border(z)$ between $v$ and a node $w \in S$.
\end{itemize}
Then $S \cup \{v\}$ is also communicating.
\end{theorem}

\begin{proof}
We use Lemma~\ref{lem1} and Lemma~\ref{lem2}, detailed below. According to these lemmas:
\begin{itemize}
\item $\forall x \in S$, $v$ eventually accepts $(x,x.m_{0})$ (Lemma~\ref{lem1})
\item $\forall x \in S$, $x$ eventually accepts $(v,v.m_{0})$ (Lemma~\ref{lem2})
\end{itemize}
Then, according to Definition~\ref{defcom}, $S \cup \{v\}$ is communicating.

\end{proof}

\begin{lemma} \label{lem1} Let there be the same hypothesis as Theorem~\ref{thcom}.
Then $\forall x \in S$, $v$ eventually accepts $(x,x.m_{0})$.
\end{lemma}

\begin{proof} Let $x$ be a node of $S$.

As $S$ is communicating, $u$ eventually accepts $(x,x.m_{0})$ in EXIT, and sends it to its neighbors.
So according to (1), $v$ receives $(x,x.m_{0})$ from $u$.
To accept it, according to the condition of EXIT, $v$ only needs the authorizations $(x,x.m_{0},z)$ with $z \in Z$ (possibly less, if $x \in core(z)$). Let us show that $v$ eventually receives these authorizations.

Let $z \in Z$ be. Then, according to (2), there exists a correct path on its border between $v$ and a node $w \in S$.
As $S$ is communicating, $w$ eventually accepts $(x,x.m_{0})$. So, according to EXIT, $w$ sends $(x,x.m_{0},z)$ to its neighbors. So does each node of the correct path in DIFF, until $(x,x.m_{0},z)$ reaches $v$. Thus, the result.
\end{proof}

\begin{lemma} \label{lem2} Let there be the same hypothesis as Theorem~\ref{thcom}.
Then $\forall x \in S$, $x$ eventually accepts $(v,v.m_{0})$.
\end{lemma}

\begin{proof} Let $x$ be a node of $S$.

First, let us establish two preliminary results:
\begin{itemize}
\item (a) Initially, $v$ sends $(v,v.m_{0})$. So $u$ receives $(v,v.m_{0})$ from $v$. According to ENTER, $(v,v.m_{0},v)$ is added to $u.Wait$. As there is no zone $z$ such that $v \in core(z)$ and $v \notin core(z)$, no authorization is required in EXIT. Thus, $u$ eventually accepts $(v,v.m_{0})$.

\item (b) Let $z \in Z$ be. Let $y$ be any node of $S$ in $border(z)$.
\begin{itemize}
\item As $z \in v.myCtr$, $v$ sent $(v,v.m_{0},z)$ in INIT.
\item According to (2), there exists a correct path on $border(z)$ between $v$ and $w$. 
\end{itemize}
So $w$ eventually receives the authorization $(v,v.m_{0},z)$. $S$ is a set of correct nodes, $w \in S$ and $border(z)$ is connected. Thus, there also exists a correct path on $border(z)$ between $w$ and $y$. Then $y$ eventually receives $(v,v.m_{0},z)$.
\end{itemize}

Now, let $\mathcal{C}_{1}$ be a configuration in which we have reached the states described in (a) and (b). Then, in such a configuration, $(v,v.m_{0})$ becomes \emph{indistinguishable} from $(u,u.m_{0})$. Indeed, the only part of the protocol that could distinguish these messages is the condition of EXIT, for the
nodes of $border(z)$ with $z \in Z$. But, as we have reached
the state described in (b), all these nodes have received the
authorizations $(v,v.m_{0},z)$, $z \in Z$. So EXIT behaves the
same way in both cases. Thus, as $x$ eventually accepts $(u,u.m_{0})$, $x$ eventually accepts $(v,v.m_{0})$.

This notion of \emph{indistinguishability} is deliberately intuitive: the detailed proof is by exhaustion, and presents no particular interest. However, let us give the sketch of this proof.

As $x$ eventually accepts $(u,u.m_{0})$, let
$(\mathcal{A}_{1}(u), \dots , \mathcal{A}_{n}(u))$ be the list of actions related to $(u,u.m_{0})$, by order of execution. These actions can be of the following types:
\begin{itemize}
\item $p$ sends $(u,u.m_{0})$ to $q$
\item $p$ sends $(u,u.m_{0},z)$ to $q$, $z \in Ctr$
\item $p$ receives $(u,u.m_{0})$ from $q$
\item $p$ receives $(u,u.m_{0},z)$ from q, $z \in Ctr$
\item $p$ accepts $(u,u.m_{0})$ from $q$
\end{itemize}

Let there be any configuration occuring after $\mathcal{C}_{1}$. Let $k < n$ be such that:
\begin{itemize}
\item The actions $(\mathcal{A}_{1}(v),\dots,\mathcal{A}_{k}(v))$ have been executed.
\item $\mathcal{A}_{k+1}(v)$ has not been executed yet.
\end{itemize}

Then, as $\mathcal{A}_{k+1}(u)$ eventually occurs, $\mathcal{A}_{k+1}(v)$ eventually occurs by the same mechanism. Notice that $\mathcal{A}_{k+1}(v)$ can also occur by another mechanism: this one is by default. Thus, by recursion, the result.\end{proof}

\subsection{Determination of a reliable node set}

This theorem is the combination of the two previous theorems: we simply show that a \emph{safe} and \emph{communicating} node set is \emph{reliable}. In pratice, to determine a reliable node set, we just have to:
\begin{itemize}
\item Determine a safe node set $S_{1}$, if we manage to.
\item Construct a communicating node set $S_{2}$.
\item Make the intersection of $S_{1}$ and $S_{2}$.
\end{itemize}

\begin{theorem}[Determination of a reliable node set]
\label{threl}
Let $S_{1}$ be a safe node set.
Let $S_{2}$ be a communicating node set.
Then $S = S_{1} \cap S_{2}$ is a reliable node set.
\end{theorem}

\begin{proof}
Let $p$ and $q$ be two nodes of $S$.
\begin{itemize}
\item As $q$ also belongs to $S_{1}$, $q$ never accepts a false message.
\item As $p$ and $q$ also belong to $S_{2}$, $q$ eventually accepts $(p,p.m_{0})$.
\end{itemize}
Then, by definition \ref{defrel}, $S$ is reliable.
\end{proof}

\subsection{Message complexity}

\label{mcgen}
Let us evaluate the message complexity of this protocol when the whole network is a reliable node set (according to Definition~\ref{defrel}).

We define the following parameters:
\begin{itemize}
\item $n$, the number of nodes of the network.
\item $d$, the \emph{degree} of the network, that is: the maximal number of neighbors for a node.
\item $N_{Ctr}$, the number of control zones in $Ctr$.
\item $N_{Border}$, the maximal number of nodes in the border of a control zone.
\end{itemize}

Let $u$ be any node. Let us evaluate the number of messages related to $u$:
\begin{itemize}
\item All nodes once accept and send $(u,u.m_{0})$ to their neighbors, which makes at most $dn$ messages.
\item All nodes on the border of a control zone $z$ once send $(u,u.m_{0},z)$ to their neighbors, which makes at most $dN_{Border}N_{Ctr}$ messages.
\end{itemize}

Thus, at most $dn(n + N_{Border}N_{Ctr})$ messages are sent in the network. Therefore, if we assume that $N_{Border}$ is $o(1)$ and that $N_{Ctr}$ is $o(n)$, the message complexity is $o(n^{2})$, the same as a standard broadcast protocol (see~\ref{classical}).

\section{Evaluation of the Protocol}

\label{secexp}
In this section, we make an experimental evaluation of our protocol. We describe our methodology and our case of study, then comment on the results.

\subsection{Methodology}

We want to evaluate the performances of our protocol for a given network topology and a given choice of control zones. As we are most interested in obtaining a high \emph{proportion} of correct nodes which communicate reliably, assuming a probabilistic distribution of the Byzantine nodes is a sensible option. It is also motivated by real constraints of actual networks that, we believe, could be supporting our protocol: for instance, most peer-to-peer overlay networks give newcomers a random identifier and thus a random position in the virtual topology for the purpose of load balancing, and many virus propagation mechanisms use random epidemic schemes to spread across networks. The main metric for evaluating the efficiency of a reliable communication scheme is the probability for two correct nodes, selected uniformly at random, to communicate properly.

\label{pnbyz}

Our evaluation scheme has the following input and output:
\begin{itemize}
\item \emph{Input}: $n_{B}$, the number of Byzantine nodes on the network, randomly distributed.
\item \emph{Output}: $P(n_{B})$, the probability that two randomly choosen nodes communicate reliably. That is: each node eventually receives the message of the other, and never accepts any lying message.
\end{itemize}

To evaluate $P(n_{B})$, we use simulations and a Monte Carlo method \cite{MU05b}. For a given value of $n_{B}$, we run a large number of simulations. The fraction of succesful simulations converges to $P(n_{B})$. It is actually impossible to simulate a distributed algorithm in the presence of Byzantine failures. Indeed, it would imply to predict the worst possible behavior of Byzantine nodes, which is a far too difficult problem. Also, giving a specific behavior to faulty nodes would weaken the problem tremendously. Therefore, instead of simulating the protocol, we use the theorems of Section~\ref{secthm}.

Here are the main steps of a single simulation:
\begin{itemize}
\item Among the nodes of the network, choose $n_{B}$ nodes that are Byzantine, uniformly at random. That is, all possible distributions of Byzantine nodes have the same probability to occur.
\item Use Theorem~\ref{threl} to construct a \emph{reliable node set} (see Definition~\ref{defrel}). In the worst case, this set is empty. An illustration of this step is given in \ref{toyex}.
\item Choose two nodes uniformly at random. If both nodes are in the reliable node set, the simulation is a success. Else, it is a failure.
\end{itemize}

Notice that we only construct \emph{one} reliable node set, which may not necessarily be the best one. Therefore, the estimated value of $P(n_{B})$ is a lower bound of the real value of $P(n_{B})$. This is not a problem, as we only want to give guarantees.

\subsection{Topology and Control Zones}

In this subsection, we choose an particular network topology and a particular set of control zones to perform our evaluation. Then we present an example of the construction of a reliable node set.

\subsubsection{Network Topology}

We consider a \emph{torus} network and a \emph{grid} network. A torus network can be seen as a grid network with a continuity between the left-right and up-down extremities.
Our topology choice is motivated by the fact that those are the most simple yet two-dimensional topologies with fixed degree. 

\begin{definition}[Torus and grid]
A $N \times N$ \emph{torus} (resp. \emph{grid}) network is a network such that:
\begin{itemize}
\item Each node has a unique identifier $(i,j)$ with
$1 \leq i \leq N$ and $1 \leq j \leq N$.
\item Two nodes $(i_{1},j_{1})$ and $(i_{2},j_{2})$ are neighbors if and only if one of these two conditions is satisfied:
\begin{itemize}
\item $i_{1} = i_{2}$ and $ \lvert j_{1}-j_{2} \rvert  = 1$ or $N$ (resp. $1$).
\item $j_{1} = j_{2}$ and $ \lvert i_{1}-i_{2} \rvert = 1$ or $N$ (resp. $1$).
\end{itemize}
\end{itemize}
\end{definition}

\begin{figure*}
\begin{center}
\includegraphics[width=12cm]{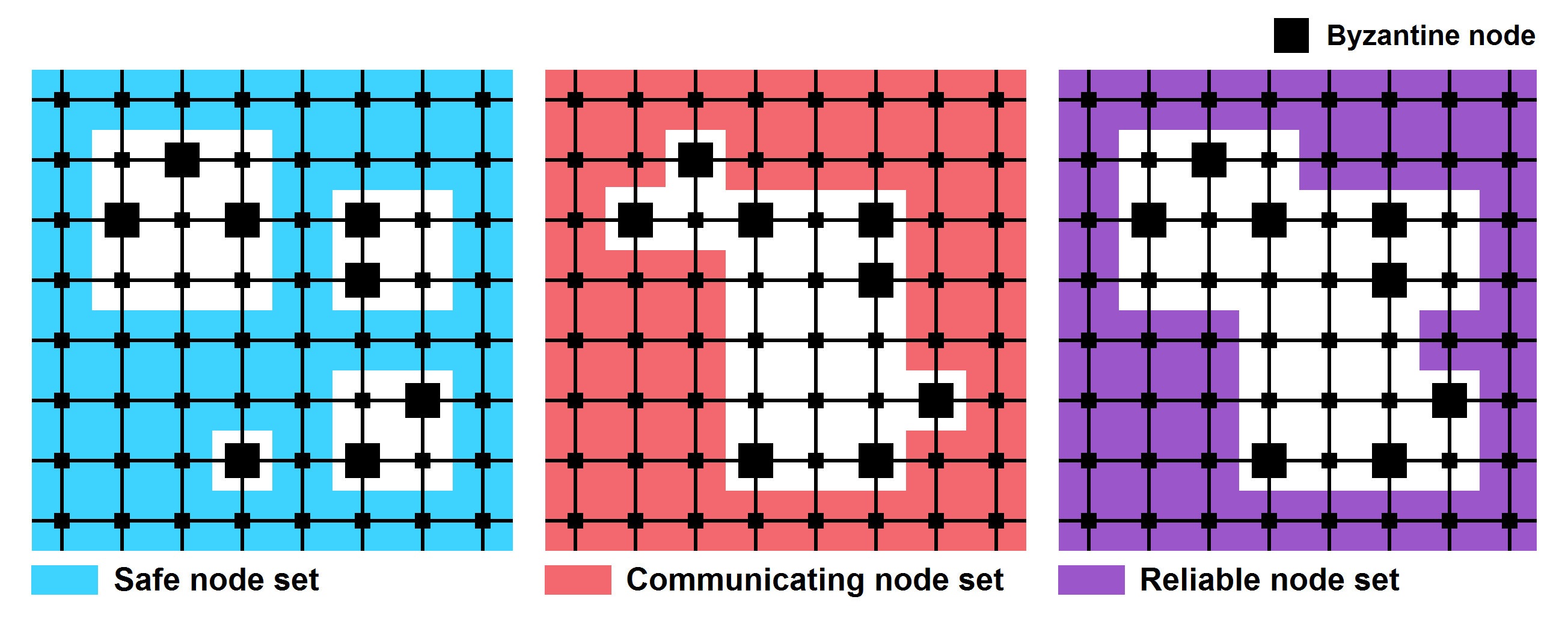}
\caption{Toy example of safe, communicating and reliable node sets} 
\label{fig:toyex}
\end{center}
\end{figure*}

\subsubsection{Control Zones}

Our protocol is given a set of \emph{square control zones}. See \ref{mydefctr} for the definition of a control zone.

\begin{definition}[Square control zone]
A \emph{square control zone} $Sqr(i_{0},j_{0},w)$ is a pair $(Core, Border)$ such that:

\begin{itemize}
\item $Core$ if the set of nodes $(i,j)$ such that:
\begin{itemize}
\item $i_{0} < i < i_{0} + w + 1$
\item $j_{0} < j < j_{0} + w + 1$
\end{itemize}
\item $Border$ is the set of nodes $(i,j)$ such that:
\begin{itemize}
\item  $i = i_{0}$ or $i_{0} + w + 1$,
and $j_{0} \leq j \leq j_{0} + w + 1$
\item  $j = j_{0}$ or $j_{0} + w + 1$,
and $i_{0} \leq i \leq j_{0} + w + 1$
\end{itemize}
\end{itemize}
The parameter $w$ denotes the \emph{width} of the control zone.

\end{definition}

For instance, Figure~\ref{fig:ctrzone} represents a square control zone of width $3$;
$(i_{0},j_{0})$ corresponds to the upper-left node of the border of the control zone.

The set of control zones $Ctr$ used by the protocol (see~\ref{mydefctr}) has parameter $W$, the \emph{order} of the protocol.

\begin{definition}[Order of the protocol]
\label{deforder}
A protocol of \emph{order} $W$ on a torus network is defined by the following set $Ctr$ of control zones:
\begin{itemize}

\item $Ctr(W) = \bigcup_{1 \leq w \leq W, 1 \leq i_{0} \leq N, 1 \leq j_{0} \leq N}$ $Sqr(i_{0},j_{0},w)$
\end{itemize}
\end{definition}

For instance, $Ctr(3)$ is the set of all square control zones of width $1$, $2$ or $3$. Notice that a node knowing its identifier $(i,j)$ and the order $W$ of the protocol can easily determine its set $myCtr$ (see \ref{mymyctr}), without any offline implementation.

\subsubsection{Message complexity}

Let us evaluate the exact number of messages sent, when the whole network is a reliable node set.

Let $n$ be the number of nodes of the network, and $W$ the order of the protocol (according to Definition~\ref{deforder}). A square control zone of width $w$ has $4(w+1)$ nodes on its border. According to Definition~\ref{deforder}, there are $n$ square control zones of a given width $w$.
Let $u$ be any node. Let us evaluate the number of messages related to $u$. First, all nodes once accept and send $(u,u.m_{0})$ to their $4$ neighbors, which makes $4n$ messages. Second, all nodes on the border of a control zone $z$ once send $(u,u.m_{0},z)$ to their $4$ neighbors, which makes $4n \sum_{w=1}^{W}4(1+w) = 8nW(W+3)$ messages.
Thus, $4n^{2}$ standard messages are sent, the same as a standard broadcast protocol (see~\ref{classical}). In addition, $8W(W+3)n^{2}$ authorization messages are sent. However, in practice, they can contain a hash code of the authorized message: they are potentially way lighter.

\subsubsection{Example of construction of a reliable node set}

\label{toyex}
The main step of a simulation is the construction of a \emph{reliable node set}, if it exists. Figure~\ref{fig:toyex} represent a toy example (extracted from a supposedly larger network), for a protocol of order $3$. Let us comment on this figure.

First, we determine a \emph{safe node set} (see Definition~\ref{defsafe}), using Theorem~\ref{thsafe}. There actually exists a set $Z$ of square control zones satisfying the conditions of Theorem~\ref{thsafe}. The blank squares correspond to the cores of the control zones of $Z$. These cores contain all Byzantine nodes, and do not intersect the union of the corresponding borders. Here, having control zones of width $3$ is an advantage: with width $2$, the upper-left group of Byzantine nodes could not be neutralized.

Then, we construct a \emph{communicating node set} $S$ (see Definition~\ref{defcom}), using Theorem~\ref{thcom}. We notice that the correct nodes surrounded by too many Byzantine nodes cannot be added to $S$, as they do not satisfy the conditions of Theorem~\ref{thcom}. Here, having control zones of width $3$ is a drawback: the presence of these zones make the conditions of Theorem~\ref{thcom} harder to satisfy, which limits the size of the communicating node set.

Finally, we determine the \emph{reliable node set} (see Definition~\ref{defrel}), using Theorem~\ref{threl}. According to this theorem, we simply take the intersection of the reliable and communicating node sets, determined previously.

\subsection{Experimental results}

In Figure~\ref{fig:influ}, we represented the influence of the number of Byzantine nodes and of the order of the protocol.
The left column gives the results for a $100 \times 100$ torus network.
The right column gives the results for the corresponding grid network, obtained by an edge-cut on the torus.
This edge-cut separates some control zones in $2$ or $4$ new zones. We consider that these new control zones have independent identifiers.

\begin{figure*}
\begin{center}
\includegraphics[width=18cm]{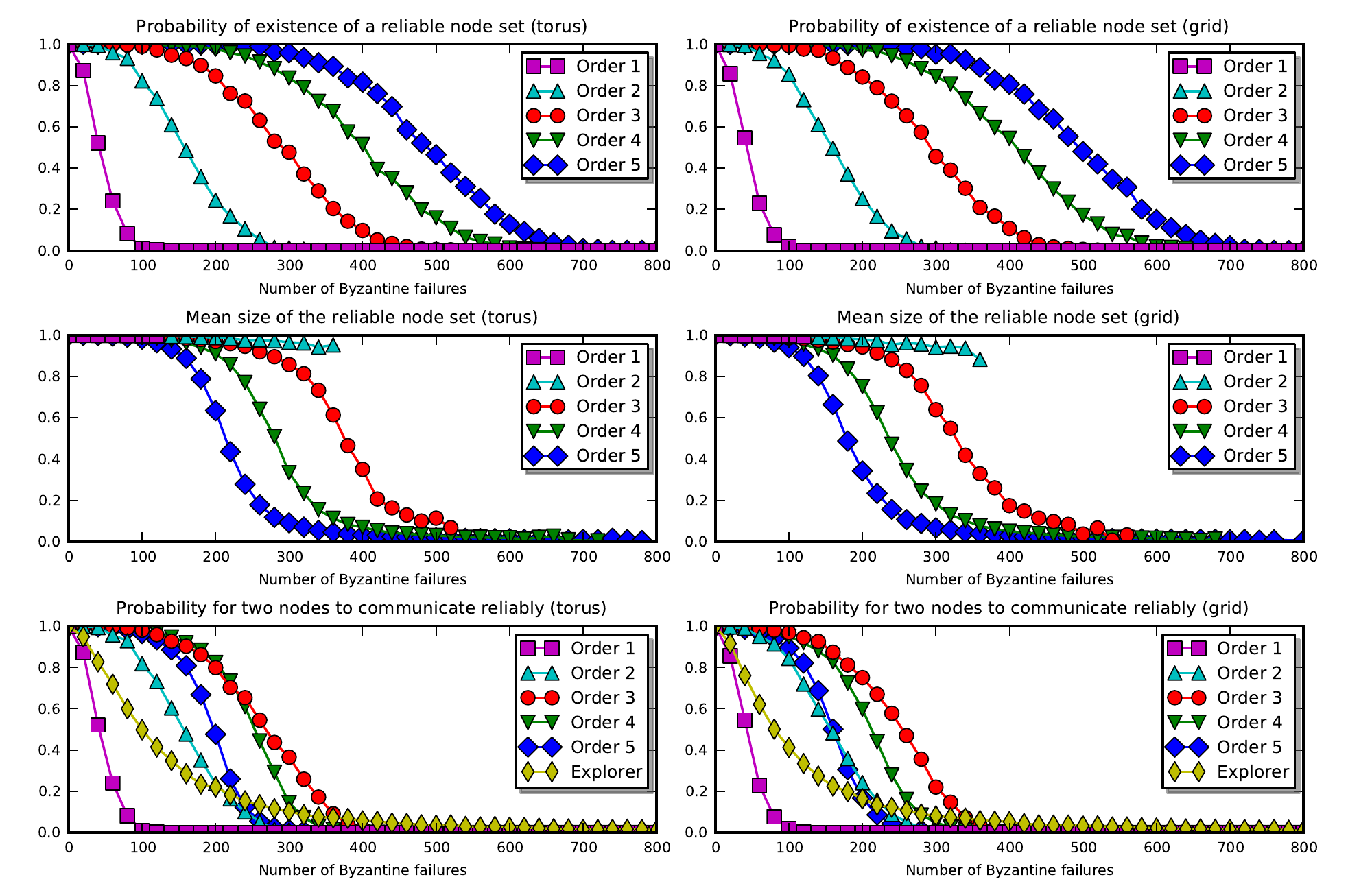}
\caption{Influence of the parameters of simulation} 
\label{fig:influ}
\end{center}
\end{figure*}

The first row of Figure~\ref{fig:influ} represents the probability that a reliable node set \emph{exists}. As any correct node forms a communicating node set, this corresponds to the probability of existence of a safe node set. That is, the probability that a set $Z$ of control zones satisfies the conditions of Theorem~\ref{thsafe}.
Figure~\ref{fig:influ} shows that this probability \emph{increases} with the order.
To give an illustration: a group of $3 \times 3$ Byzantine nodes cannot be neutralized at order $2$. But with the new zones added at order $3$, it becomes possible. And so forth. 

The second row of Figure~\ref{fig:influ} represents the mean size of the reliable node set, when it exists. That is, the fraction of correct nodes covered by the reliable node set.
Figure~\ref{fig:influ} shows that this fraction \emph{decreases} with the order. Indeed, increasing the order makes the conditions of Theorem~\ref{thcom} harder to satisfy, as the set $Z$ of Theorem~\ref{thcom} contains more zones. Therefore, the construction of the communicating node set is made more difficult.

The third row of Figure~\ref{fig:influ} represents the probability  $P(n_{B})$ for two correct nodes, choosen uniformly at random, to communicate reliably. In the previous plots, we showed that the order of the protol had a \emph{positive} influence on the existence of a reliable node set, but a \emph{negative} influence on its size. Therefore, an compromise between these two tendancies appears for order $3$, for which the probability is optimal.

It is difficult to compare our proposal with the high-connectivity based approaches ~\cite{D82j,NT09j} since they always assume the worst possible placement of Byzantine nodes. For example, Explorer~\cite{NT09j} has the sender use all paths and the receiver validates the first message coming from $2k+1$ node-disjoint paths by a majority. So an asynchronous schedule in a torus may always choose that node-disjoint paths going through Byzantine nodes arrive first, defeating the protocol anytime the number of Byzantine is greater than 1. In order to maximize the efficiency of Explorer in our context (\emph{i.e.} assuming random placement of the nodes), we force the sender to use exactly $4$ node-disjoint fixed paths. Then, for a particular sender and receiver, the Explorer protocol works if and only if Byzantine nodes are located on at most one of those four paths. We present the corresponding guarantees of the modified Explorer protocol in Figure~\ref{fig:influ}.
It turns out that our protocol outperforms the modified Explorer by a significant margin. For example, if the goal probability is $P(n_{B}) \geq 0.99$, then on the grid (resp. torus) topology, the modified version of Explorer can tolerate at most $5$ (resp. $7$) Byzantine nodes. Our approach can tolerate up to $50$ (resp. $80$) Byzantine nodes.

\section{Conclusion}

In this paper, we showed that, if we accept that a small minority of correct nodes does not communicate reliably in the presence of Byzantine failures, we make it possible to tolerate a large number of those failures, even in a low-connectivity network. We proposed an experimental methodology to obtain probabilistic guarantees, and illustrated this on torus and grid shaped networks, with an uniform distribution of Byzantine failures. Yet, the same principle could be applied to any network and any distribution of Byzantine failures. However, it requires to define a proper set of control zones to limit the Byzantine influence. Defining optimal sets of control zones for a given communication graph and Byzantine node distribution is a challenging open question.

\bibliographystyle{plain}
\bibliography{biblio}

\begin{thebibliography}{10}

\bibitem{AW98b}
H.~Attiya and J.~Welch.
\newblock {\em Distributed Computing: Fundamentals, Simulations, and Advanced
  Topics}.
\newblock McGraw-Hill Publishing Company, New York, May 1998.
\newblock 6.

\bibitem{BV05c}
Vartika Bhandari and Nitin~H. Vaidya.
\newblock On reliable broadcast in a radio network.
\newblock In Marcos~Kawazoe Aguilera and James Aspnes, editors, {\em PODC},
  pages 138--147. ACM, 2005.

\bibitem{CL99c}
Miguel Castro and Barbara Liskov.
\newblock Practical byzantine fault tolerance.
\newblock In {\em OSDI}, pages 173--186, 1999.

\bibitem{D82j}
D.~Dolev.
\newblock The {Byzantine} generals strike again.
\newblock {\em Journal of Algorithms}, 3(1):14--30, 1982.

\bibitem{DFS05c}
Vadim Drabkin, Roy Friedman, and Marc Segal.
\newblock Efficient byzantine broadcast in wireless ad-hoc networks.
\newblock In {\em DSN}, pages 160--169. IEEE Computer Society, 2005.

\bibitem{DMT10ca}
Swan Dubois, Toshimitsu Masuzawa, and S\'{e}bastien Tixeuil.
\newblock The impact of topology on byzantine containment in stabilization.
\newblock In {\em Proceedings of DISC 2010}, {L}ecture {N}otes in {C}omputer
  {S}cience, Boston, Massachusetts, USA, September 2010. {S}pringer {B}erlin /
  {H}eidelberg.

\bibitem{DMT10cd}
Swan Dubois, Toshimitsu Masuzawa, and S\'{e}bastien Tixeuil.
\newblock On byzantine containment properties of the min+1 protocol.
\newblock In {\em Proceedings of SSS 2010}, {L}ecture {N}otes in {C}omputer
  {S}cience, New York, NY, USA, September 2010. {S}pringer {B}erlin /
  {H}eidelberg.

\bibitem{DMT11j}
Swan Dubois, Toshimitsu Masuzawa, and S\'{e}bastien Tixeuil.
\newblock Bounding the impact of unbounded attacks in stabilization.
\newblock {\em IEEE Transactions on Parallel and Distributed Systems (TPDS)},
  2011.

\bibitem{DMT11cb}
Swan Dubois, Toshimitsu Masuzawa, and S{\'e}bastien Tixeuil.
\newblock Maximum metric spanning tree made byzantine tolerant.
\newblock In David Peleg, editor, {\em Proceedings of DISC 2011}, Lecture Notes
  in Computer Science (LNCS), Rome, Italy, September 2011. Springer Berlin /
  Heidelberg.

\bibitem{K04c}
Chiu-Yuen Koo.
\newblock Broadcast in radio networks tolerating byzantine adversarial
  behavior.
\newblock In Soma Chaudhuri and Shay Kutten, editors, {\em PODC}, pages
  275--282. ACM, 2004.

\bibitem{LSP82j}
Leslie Lamport, Robert~E. Shostak, and Marshall~C. Pease.
\newblock The byzantine generals problem.
\newblock {\em ACM Trans. Program. Lang. Syst.}, 4(3):382--401, 1982.

\bibitem{MMR03j}
D.~Malkhi, Y.~Mansour, and M.K. Reiter.
\newblock Diffusion without false rumors: on propagating updates in a
  {Byzantine} environment.
\newblock {\em Theoretical Computer Science}, 299(1--3):289--306, April 2003.

\bibitem{MRRS01c}
D.~Malkhi, M.~Reiter, O.~Rodeh, and Y.~Sella.
\newblock Efficient update diffusion in byzantine environments.
\newblock In {\em The 20th {IEEE} Symposium on Reliable Distributed Systems
  ({SRDS} '01)}, pages 90--98, Washington - Brussels - Tokyo, October 2001.
  IEEE.

\bibitem{MT06cb}
Toshimitsu Masuzawa and S\'{e}bastien Tixeuil.
\newblock Bounding the impact of unbounded attacks in stabilization.
\newblock In Ajoy~Kumar Datta and Maria Gradinariu, editors, {\em SSS}, volume
  4280 of {\em Lecture Notes in Computer Science}, pages 440--453. Springer,
  2006.

\bibitem{MT07j}
Toshimitsu Masuzawa and S\'{e}bastien Tixeuil.
\newblock Stabilizing link-coloration of arbitrary networks with unbounded
  byzantine faults.
\newblock {\em International Journal of Principles and Applications of
  Information Science and Technology (PAIST)}, 1(1):1--13, December 2007.

\bibitem{MS03j}
Y.~Minsky and F.B. Schneider.
\newblock Tolerating malicious gossip.
\newblock {\em Distributed Computing}, 16(1):49--68, 2003.

\bibitem{MU05b}
M.~Mitzenmacher and E.~Upfal.
\newblock {\em Probability and Computing}.
\newblock Cambridge University Press, 2005.

\bibitem{NA02c}
Mikhail Nesterenko and Anish Arora.
\newblock Tolerance to unbounded byzantine faults.
\newblock In {\em 21st Symposium on Reliable Distributed Systems (SRDS 2002)},
  pages 22--29. IEEE Computer Society, 2002.

\bibitem{NT09j}
Mikhail Nesterenko and S\'{e}bastien Tixeuil.
\newblock Discovering network topology in the presence of byzantine nodes.
\newblock {\em IEEE Transactions on Parallel and Distributed Systems (TPDS)},
  20(12):1777--1789, December 2009.

\bibitem{PP05j}
Andrzej Pelc and David Peleg.
\newblock Broadcasting with locally bounded byzantine faults.
\newblock {\em Inf. Process. Lett.}, 93(3):109--115, 2005.

\bibitem{SOM05c}
Yusuke Sakurai, Fukuhito Ooshita, and Toshimitsu Masuzawa.
\newblock A self-stabilizing link-coloring protocol resilient to byzantine
  faults in tree networks.
\newblock In {\em Principles of Distributed Systems, 8th International
  Conference, OPODIS 2004}, volume 3544 of {\em Lecture Notes in Computer
  Science}, pages 283--298. Springer, 2005.

\end{thebibliography}

\end{document}